\documentclass[12pt,a4paper]{amsart}
\usepackage{mathrsfs}
\usepackage{amsfonts}
\usepackage[notref,notcite]{showkeys}
\usepackage{enumerate}

\usepackage{amsmath,amssymb,xspace,amsthm}
\newtheorem{theorem}{Theorem}%[section]
\newtheorem{example}[theorem]{Example}%[section]
\newtheorem{lemma}[theorem]{Lemma}%[section]
\numberwithin{equation}{section}

\def\span{\operatorname{span}}
\renewcommand{\H}{\mathcal{H}}
\def\lt{\operatorname{lt}}
\newtheorem{claim}{Claim}
\newtheorem{case}{Case}

\newcommand{\C}{\ensuremath{\mathbb C}\xspace}

\renewcommand{\a}{\ensuremath{\alpha}}
\renewcommand{\b}{\ensuremath{\beta}}

\newcommand{\Z}{\ensuremath{\mathbb{Z}}\xspace}
\newcommand{\N}{\ensuremath{\mathbb{N}}\xspace}

\renewcommand{\phi}{\varphi}

\begin{document}
\title[Simple modules over $\mathcal{H}(f)$]{Finite-dimensional simple modules over  generalized Heisenberg algebras}
\author{Rencai L{\"u} and Kaiming Zhao}
\date{Jan.3, 2014}
\maketitle

\begin{abstract} Generalized Heisenberg algebras $\H(f)$ for any polynomial $f(h)\in\C[h]$
have been used to explain various physical systems and many physical
phenomena for the last 20 years. In this paper, we first obtain the
center of $\H(f)$, and the necessary and sufficient conditions on
$f$ for two $\H(f)$ to be isomorphic. Then we determine all finite
dimensional simple modules over $\H(f)$ for any polynomial
$f(h)\in\C[h]$. If $f=wh+c$ for any $c\in \C$ and $n$-th ($n>1$)
primitive root $w$ of unity we actually obtain a complete
classification of all irreducible modules over $\mathcal{H}(f)$. For
 many $f\in\C[h]$, we also prove that,  for any $n\in\N$, $\mathcal{H}(f)$ has
infinitely many ideals $I_n$ such that $\mathcal{H}(f)/I_n\cong
M_n(\C)$, the matrix algebra.
\end{abstract}


%\vskip 5pt \noindent {\em 2000  Math. Subj. Class.:} 17B10, 17B20,
%17B65, 17B66, 17B68

\vskip 10pt

\section{Introduction}
Algebraic methods have long been applied to  solutions
of a large number of physical systems and many physical phenomena.
The most classical example of an algebraically solved system
is the harmonic oscillator, whose underlying algebra is the
Heisenberg algebra, generated by annihilation and creation operators.
Recent examples are the   concept of deformed Heisenberg algebras \cite{B, M}, that have
been used in many areas, as nuclear physics \cite{BD}, condensed matter
\cite{MRW}, atomic physics \cite{ADF}, etc.

Guided  by the wide range of physical applicability of the Heisenberg algebra there
have been a lot of successful efforts in the last 30 years to analyze possible physical relevance of q-oscillators or
deformed Heisenberg algebras \cite{AEGPL, GPLL, P}. The expected physical properties of toy systems described
by these generalized Heisenberg algebras were analyzed and indications on how to solve an
old puzzle in physics were obtained \cite{MRW}.

Sometime later,  logistic algebras, that are a generalization of the
Heisenberg algebra  were introduced. These algebras have finite and
infinite dimensional representations associated with the cycles of
the logistic map and infinite-dimensional representations related to
the chaotic band \cite{CR1}. A quantum solid Hamiltonian whose
collective modes of vibration are described by oscillators
satisfying the logistic algebra was constructed and the
thermodynamic properties of this model in the two-cycle and in a
specific chaotic region of the logistic map were analyzed. In the
chaotic band this model shows  how a quantum system can present a
nonstandard quantum behaviour \cite{CR1}.

In 2001, as a generalization of logistic algebras, the concept of
generalized Heisenberg algebras were formally introduced in
\cite{CR2}. Applications of this algebraic structure to different
systems have appeared in the literature in the last few years
\cite{BCR1, BCR2, BEH, CRRH,  CRRL, DCR, DOR, HCR}. In particular,
how catlike superpositions of generalized Heisenberg algebras
nonlinear coherent states behave under dissipative  decoherence was
studied in \cite{CRRL}; coherent states for power-law potentials
were constructed using generalized Heisenberg algebra in \cite{BEH};
the anharmonic spectrum of diatomic molecules was reproduced based
on a generalized Heisenberg algebra in \cite{DOR}. In \cite{CHR},
following the harmonic oscillator scheme, the authors presented a
realization of  generalized Heisenberg algebra for the free particle
in an infinite square-well potential. Now let us recall the
generalized Heisenberg algebras.

We will denote by $\mathbb{Z}$, $\mathbb{Z}_+$, $\N$, and $\mathbb{C}$
the sets of  all integers, nonnegative integers, positive integers,
 and complex numbers, respectively. For a Lie algebra
$L$ we denote by $U(L)$ the universal enveloping algebra of $L$.

\vskip 5pt For any  $f(h)\in\C[h]$, the {\it generalized Heisenberg
algebra} $\mathcal{H}(f)$   is the associative algebra over $\C$
generated by $x, y, h$ with definition relations:
$$hx = xf(h); \,\,\, yh = f(h)y; \,\,\,[y,x] = f(h)-h.$$
The original definition for $\mathcal{H}(f)$ was using any analytic function $f(h)$. In
 the present paper we only study generalized Heisenberg algebras $\mathcal{H}(f)$ for $f(h)\in\C[h]$.
It is easy to see that  $z = xy-h = yx-f(h)$ is a central element in $\mathcal{H}(f)$.

 The algebra $\mathcal{H}(f)=\C[x, h, y]$
is  commutative if and only if  $f(h) = h$. It is easy to see that
$\mathcal{H}(f)$ is not a Smith algebra \cite{S} whose isomorphism
was solved in \cite{BJ}, or a generalized Weyl algebra \cite{Ba}.
%Thus, from now on we assume that $f(h)\ne h-a$ for any $a\in \C$.

For convenience,  we define $f^{(1)}(h)=f(h),
f^{(2)}(h)=f(f^{(1)}(h))$ and $f^{(i+1)}(h)=f(f^{(i)}(h))$ for all
$i\in\N$.

A module $V$ over $\mathcal{H}(f)$ is called a {\it weight module} if $V=\oplus_{\lambda\in\C}V_\lambda$ where
$V_\lambda=\{v\in V\,\,|\,\,hv=\lambda v\}$.

Some special finite dimensional weight modules over $\mathcal{H}(f)$ were constructed in \cite{CR2}.

The present paper is organized as follows. In Sect.2, we  obtain the
center of $\H(f)$ (Theorem 4), and the necessary and sufficient
conditions on $f$ for two $\H(f)$ to be isomorphic (Theorem 5). In
Sect.3,  we determine all finite dimensional simple modules over
$\H(f)$ for any polynomial $f(h)\in\C[h]$ (Theorem 12). These
modules are all weight modules with weight multiplicity $1$. If
$f=wh+c$ for any $c\in \C$ and for any $n$-th ($n>1$) primitive root
$w$ of unity we have actually the complete classification of all
irreducible modules over $\mathcal{H}(f)$ (Example 13 (4)).

An interesting property of many generalized Heisenberg algebras
$\mathcal{H}(f)$ is that for any $n\in\N$, $\mathcal{H}(f)$ has
infinitely many ideals $I_n$  such that $\mathcal{H}(f)/I_n\cong
M_n(\C)$, the matrix algebra (Example 15).

\section{Isomorphism classes of generalized Heisenberg algebras}

We will first study properties of generalized Heisenberg algebras
$\mathcal{H}(f)$. Then we determine the center of $\mathcal{H}(f)$,
and isomorphism classes for all generalized Heisenberg algebras. The
following lemma is a crucial result for this paper.

\begin{lemma}\label{lemma-1} The set $\{x^ih^jy^k|i,j,k\in \Z_+\}$ is a basis of $\mathcal{H}(f)$.\end{lemma}

\begin{proof} We will first show that $S=\{x^ih^jy^k|i,j,k\in \Z_+\}$ spans $\mathcal{H}(f)$.
It suffices to show that  the product $x^{i_1}h^{j_1}y^{k_1}x^{i_2}h^{j_2}y^{k_2}$ is a linear combination of elements in $S$
for any $i_1, j_1, k_1, i_2, j_2, k_2\in\Z_+$. We do this by induction on $k_1$.

If $k_1=0$, we see that
$$x^{i_1}h^{j_1}y^{k_1}x^{i_2}h^{j_2}y^{k_2}=x^{i_1+i_2}(f^{(i_2)}(h))^{j_1}
h^{j_2}y^{k_2}\in\span(S).$$ Similarly, if $i_2=0$,  then
$x^{i_1}h^{j_1}y^{k_1}x^{i_2}h^{j_2}y^{k_2}\in\span(S).$ By
induction we suppose
$x^{i_1}h^{j_1}y^{k_1}x^{i_2}h^{j_2}y^{k_2}=x^{i_1+i_2}
(f^{(i_2)}(h))^{j_1} h^{j_2}y^{k_2}\in\span(S)$ for a fixed $k_1
\in\Z_+$, and for all $i_1, j_1, i_2, j_2, k_2\in\Z_+$.

If $i_2=0$, we knew that $$x^{i_1}h^{j_1}y^{k_1+1}x^{i_2}h^{j_2}y^{k_2}\in\span(S).$$
Now suppose that $i_2>0$, and we have $$x^{i_1}h^{j_1}y^{k_1+1}x^{i_2}h^{j_2}y^{k_2}=x^{i_1}h^{j_1}y^{k_1}yxx^{i_2-1}h^{j_2}y^{k_2}$$
$$=x^{i_1}h^{j_1}y^{k_1}(xy+f(h))x^{i_2-1}h^{j_2}y^{k_2}$$
$$\equiv x^{i_1}h^{j_1}y^{k_1}xyx^{i_2-1}h^{j_2}y^{k_2}\mod \span(S)$$
$$\equiv ...\equiv x^{i_1}h^{j_1}y^{k_1} x^{i_2}yh^{j_2}y^{k_2}\mod \span(S)$$
$$\equiv x^{i_1}h^{j_1}y^{k_1} x^{i_2}(f(h))^{j_2}y^{k_2+1}\mod \span(S)$$
$$\equiv 0\mod \span(S).$$
Thus $ \span(S)=\mathcal{H}$.

Next we prove that $S$ is linearly independent. Let $V$ be the polynomial algebra $\C[X, H, Y]$.
Let us define the action of $\H$ on $V$ as follows:
$$x(X^iH^jY^k)=X^{i+1}H^jY^k,$$
$$h(X^iH^jY^k)=X^if^{(i)}(H)H^jY^k;$$
and define the action of $y$ by induction on $i$ as follows:
$$y(H^jY^k)=f(H)^jY^{k+1},\,\,\,$$
%$$ y(XH^jY^k)=(xy+f(h)-h)H^jY^k=Xf(H)^jY^{k+1}+(f(h)-h)H^jY^k),$$
$$ y(X^iH^jY^k)=(xy+f(h)-h)X^{i-1}H^jY^k, \forall i\in\N.$$
Now we show that the above action of $\H$ on $V$ makes $V$ into a
module over the associative algebra $\H$. It is straightforward to
verify that, for all $i,j,k\in\Z_+$,
$$hx(X^iH^jY^k)=xf(h)(X^{i}H^jY^k),\,\,\, $$ $$(yx-xy)(X^iH^jY^k)=(f(h)-h)(X^{i}H^jY^k),$$
$$yh(H^jY^k)=f(h)y(H^jY^k).$$
Using the above established formulas, now  we prove
$yh(X^iH^jY^k)=f(h)y(X^{i}H^jY^k)$ by induction on $i\in\Z_+$:
$$\aligned &yh(X^{i+1}H^jY^k)=y(X^{i+1}f^{(i+1)}(H)H^jY^k)\\
=&(xy+f(h)-h)(X^{i}f^{(i+1)}(H)H^jY^k)\\
=&xyf(h)(X^{i}H^jY^k)+(f(h)-h)(X^{i}f^{(i+1)}(H)H^jY^k)\\
=&xf^{(2)}(h)y(X^{i}H^jY^k)+(f(h)-h)(X^{i}f^{(i+1)}(H)H^jY^k)\\
=&f(h)xy(X^{i}H^jY^k)+f(h)(f(h)-h)(X^{i}H^jY^k)\\
=&f(h)(xy+f(h)-h)(X^{i}H^jY^k)\\
 =&f(h)y(X^{i+1}H^jY^k).\endaligned$$
Thus $V$ is an $\H$-module.

By considering the action of elements of $S$ on $1\in V$, we see
that $S$ is linearly independent. The lemma follows.
\end{proof}

Denote the lexicographical order on $\Z_+^2$ by $(i,j)>(i',j')$,
that means $i>i'$, or $i=i'$ and $j>j'$. For any $0\ne \a\in
\mathcal{H}(f)$, from Lemma \ref{lemma-1}, we may uniquely write
$$\a=x^ng(h)y^m+\sum_{(i,j)<(n,m)}x^ig_{ij}(h)y^j,$$ where $0\ne
g(h)\in \C[h]$ and $g_{ij}(h)\in \C[h]$ for all $i,j$. Denote the
degree and the leading term of $\a$ by $\deg \a=(n,m)$ and
lt$(\a)=x^ng(h)y^m$ respectively.

\begin{lemma}\label{lemma-2}  \begin{itemize}
\item[(1).]  If $f\not\in \C$, then $\deg(\a\b)=\deg(\a)+\deg(\b)$
for any nonzero $\a,\b\in \mathcal{H}(f)$.
\item[(2).]  The algebra $\mathcal{H}(f)$ has no zero-divisors if and only if $f\not\in
\C$.\end{itemize}
\end{lemma}

\begin{proof} (1). Suppose that $f\not\in \C$. Say $\lt(\a)=x^ng(h)y^m$ and $\lt(\b)=x^s g_1(h)y^k$. Then it is straightforward to verify that $$\lt(\a\b)=x^{n+s}g(f^{(s)}(h))g_1(f^{(m)}(h))y^{k+m}.$$

(2).  The ``if part" follows from (1). The ``only if part" follows from the fact that $(h-f)x=0$ if $f\in \C$.\end{proof}

\begin{lemma}\label{lemma-3} \begin{itemize}
\item[(1).]  Let $f_1,f_2\in \C[h]$ such  that $f_2(ah+c)=af_1(h)+c$ for some $a\in \C^*, c\in \C$.
Then  $\mathcal{H}(f_1)\cong \mathcal{H}(f_2).$
\item[(2).] If $f\in \C$,
then the center $Z(\mathcal{H}(f))$ of $\mathcal{H}(f)$ is
$\C[z]$.\end{itemize}
\end{lemma}
\begin{proof}(1). It is straightforward to verify the following equations in $\mathcal{H}(f_1)$
$$(ah+c)(ax)=(ax)(af_1(h)+c)=(ax)f_2(ah+c),$$
 $$y(ah+c)=(af_1(h)+c)y=f_2(ah+c)y,$$
 $$[(ax),y]=(ah+c)-(af_1(h)+c)=(ah+c)-f_2(ah+c).$$  Hence from Lemma 1 and the definition of $\mathcal{H}(f_2)$, there exists a unique algebra isomorphism $\sigma:\mathcal{H}(f_2)\rightarrow \mathcal{H}(f_1)$ with $\sigma(x)=ax,\sigma(h)=ah+c$ and $\sigma(y)=y$.

(2). From (1), we know that $\mathcal{H}(f)\cong \mathcal{H}(0)$ for any $f\in\C$. We only need to consider the case that $f=0$. In this case we have $z=xy-h=yx$, $hx=yh=0$. By induction on $j$ we deduce that  $$h^j=(xy-z)^j=(-1)^jz^{j-1}(xy-z),\,\,\forall \,\, j\in\N.$$ From Lemma 1 we see that $\{x^iy^jz^k\,\,|\,\,i,j,k\in\Z_+\}$ is a basis of $\mathcal{H}(0)$.
Let
 $$\a=x^ny^mg(z)+\sum_{(i,j)<(n,m)}x^iy^jg_{ij}(z)\in \mathcal{H}(f)\backslash \C[z].$$ Without lose of generality, we may assume that $n>0$. Then from $[y,x^ky^iz^j]=x^{k-1}y^iz^{j+1}-x^ky^{i+1}z^j$ for all $i,j\in \Z_+$ and $k\in \N$. By computing the terms $x^ny^{m+1}$ under the basis $\{x^iy^jz^k|i,j,k\in \Z_+\}$, we have $[y,\a]\ne 0$. Thus $\a\not\in Z(\mathcal{H}(f))$.
\end{proof}

Now we can determine the center of $\mathcal{H}(f)$.

\begin{theorem}\label{thm-4}  If $f(h)=wh+(1-w)c$ for some $l$-th primitive root $w$ of unity and $c\in \C$, then
$Z(\mathcal{H}(f))=\C[x^l,y^l,(h-c)^l,z]$. Otherwise, $Z(\mathcal{H}(f))=\C[z]$.
\end{theorem}

\begin{proof} Let $R(f)=\C[x^l,y^l,(h-c)^l,z]$ if $f(h)=wh+(1-w)c$ for some $l$-th primitive
root $w$ of unity and $c\in \C$, and $R(f)=\C[z]$ otherwise. It is
straightforward to verify that $R(f)\subseteq Z(\mathcal{H}(f))$.
Now suppose that $R(f)\ne Z(\mathcal{H}(f))$. From Lemma
\ref{lemma-3} (2), we may assume that $f(h)\not\in \C$. Then it is
easy to see that $Z(\mathcal{H}(f))\cap \C[h]=\C[(h-c)^l]$ if
$f(h)=wh+(1-w)c$ for some $l$-th primitive root $w$ of unity and
$c\in \C$, and $Z(\mathcal{H}(f))\cap \C[h]=\C$ otherwise.
 Let
 $$\a=x^ng(h)y^m+\sum_{(i,j)<(n,m)}x^ig_{ij}(h)y^j\in Z(\mathcal{H}(f))\backslash R(f)$$ with the minimal $(n,m)$. Without lose of generality, we may assume that $n>0$ (similarly for $m>0$). From $$0=[\a,y]=x^{n}g(h)y^{m+1}-x^{n}g(f(h))y^{m+1}+\sum_{(i,j)<(n,m+1)}x^i \bar{g}_{i,j}(h)y^j,$$ we have $g(h)=g(f(h))$, yielding that $g(h)\in Z(\mathcal{H}(f))$. So $g(h)\in \C[(h-c)^l]$  if $f(h)=wh+(1-w)c$ for some $l$-th primitive root $w$ of unity and $c\in \C$, and $g\in \C$ otherwise.

 By computing $[\a,h]=0$, we have $f^{(n)}(h)=f^{(m)}(h)$, which implies $m-n\in \Z l$ if $f(h)=wh+(1-w)c$ for some $l$-th primitive root $w$ of unity and $c\in \C$, and $m=n>0$ otherwise.
Now it is easy to see that $\a-g x^{n-m}z^{m}\in Z(\mathcal{H}(f))\backslash R(f)$ if $n\ge m$ and $\a-g y^{m-n}z^{n}\in Z(\mathcal{H}(f))\backslash R(f)$ if $m>n$, which has a lower degree than $(n,m)$, a contradiction. Thus we have completed the proof.
\end{proof}

From Lemma 3 (1) we know that $\mathcal{H}(wh+(1-w)c)\cong \mathcal{H}(wh)$ for any $l$-th primitive root $w$ of unity and for all $c\in \C$.

Now we can determine  isomorphism classes for all generalized
Heisenberg algebras.

\begin{theorem}\label{thm-5}Let $f_1,f_2\in \C[h]$. Then $\mathcal{H}(f_1)\cong \mathcal{H}(f_2)$ if and only if one of the following holds:
\begin{itemize}\item[(1).] $f_1, f_2\in\C$;
\item[(2).] $f_1=f_2=h$;
\item[(3).] $f_1=h+c_1$ and $f_2=h+c_2$ for some $c_1,c_2\in \C^*$;
\item[(4).] $f_1=a_1h+b_1$ and $f_2=a_2h+b_2$, where $a_1,a_2,b_1,b_2\in \C$
with  $a_1=a_2^{\pm1}\ne 1$;
\item[(5).] $\deg f_1\ge 2$ and $f_2(h)=af_1(a^{-1}(h-c))+c$ for some $a\in
\C^*$ and $c\in \C$.
\end{itemize}
 \end{theorem}

 \begin{proof} If Condition (1), or (2), or (3),  or (4) with $a_1=a_2$, or (5)  holds, from Lemma \ref{lemma-3}
 (1) we know that $\mathcal{H}(f_1)\cong \mathcal{H}(f_2)$. Now
 suppose that $f_1=a_1h+b_1$ and $f_2=a_2h+b_2$ with $a_2=a_1^{-1}\ne1$ holds. From Lemma \ref{lemma-3}
 (1) again we know that $\mathcal{H}(a_ih+b_i)\cong \mathcal{H}(a_ih)$.
 For any $a\in \C^*$,  in $\mathcal{H}(ah)$ we have
 $$(ah) y=a^{-1} y(ah), x(ah)=a^{-1}(ah)x, [y,x]=(ah)-a^{-1}(ah).$$ From the definition
 of $\mathcal{H}(ah)$ and using Lemma 1,
 we have the associative algebra isomorphism $\tau:\mathcal{H}(a^{-1}h)\rightarrow \mathcal{H}(ah)$ with $\tau(h)=ah,\tau(x)=y,\tau(y)=x$.
Thus, if Condition (4) holds, we have again $\mathcal{H}(f_1)\cong
\mathcal{H}(f_2)$.

 Next suppose that $\tau:\mathcal{H}(f_1)\rightarrow \mathcal{H}(f_2)$ is an isomorphism. From Lemma \ref{lemma-2} (2), it is clear that $f_1\in \C$ if and only if $f_2\in \C$.
  Let $I(f)$ be the ideal of $\mathcal{H}(f)$ generated by $h-f(h)$.
  Note that $h-f(h)$ is a factor of $f^{(i)}(h)-f^{(i+1)}(h)$ for any
  $i\in\N$. We break
  the proof  into several claims.
First, using Lemma 1 we can easily verify the following properties
for $\mathcal{H}(f)$:

 \begin{claim} \begin{itemize}
\item[(1).]  $I(f)=\span\{x^ih^j(h-f(h))y^k|i,j,k\in \Z_+\}=\cap_{J\in \mathcal{J}}J$, where
$\mathcal{J}=\{$ideals $J$ of $\mathcal{H}(f)$ $|$
$\mathcal{H}(f)/J$ is commutative or $0 \}$.
\item[(2).]  The quotient $\mathcal{H}(f)/I(f)\cong \C[X,Y,H]/\C[X,Y,H](H-f(H))$, where $\C[X,Y,H]$ is
the polynomial algebra.\end{itemize}\end{claim}

 Thus, $\tau(I(f_1))=I(f_2)$ and $\mathcal{H}(f_1)/I(f_1)\cong \mathcal{H}(f_2)/I(f_2)$,
i.e., $$\frac{\C[X,Y,H]}{\C[X,Y,H](H-f_1(H))}\cong
\frac{\C[X,Y,H]}{\C[X,Y,H](H-f_2(H))}.$$ Consequently,

 \begin{claim} \begin{itemize}
\item[(1).]   $f_1\in \C$ if and only if $f_2\in \C$;
\item[(2).]   $f_1=h$ if and only if $f_2=h$;
\item[(3).]   $\deg f_1=1$ if and only if $\deg f_2=1$;
\item[(4).]  $f_1=h+c_1$ for some $c_1\in \C^*$ if and only if $f_2=h+c_2$ for some $c_2\in \C^*$.\end{itemize}
 \end{claim}

 \begin{claim}\label{claim-3} If $f_i=a_ih+b_i$ with $1\ne a_i\in \C^*$ and $b_i\in \C$ for $i=1,2$, then $a_1=a_2$ or $a_1=a_2^{-1}$.\end{claim}
 {\it Proof of Claim \ref{claim-3}}. From Lemma \ref{lemma-3} (1), we may
 assume that $b_1=b_2=0$.  From Lemma 1 and the fact  $\tau(I(a_1h))=I(a_2h)$, we
 can deduce that $\tau(h)=ah$ for some $a\in \C^*$. Using the automorphism of
 $\mathcal{H}(a_2h)$ with $h\mapsto ah, x\mapsto x, y\mapsto ay$,
 we may further assume that $\tau(h)=h$. Let $$T(f)=\{a\in \C^*| \a h=a h\a\,\,
 \mbox{for\,\, some\,\, nonzero}\,\, \a\in \mathcal{H}(f)\}.$$ Then
 $$\{a_1^i|i\in \Z\}=T(f_1)=T(f_2)=\{a_2^i|i\in \Z\}.$$ If $a_1$ is
 not a root of unity, then we have $a_1=a_2$ or $a_1=a_2^{-1}$. So we may
 assume that $a_1=w, a_2=w^p$ where $w$ is an $l$-th primitive root of unity with $l\ge 3$
 and $p,l$ coprime. Denote $$M_i=\{\a\in \mathcal{H}(a_ih)|\a h=w h\a\,\,\mbox{or}\,\,
 \a h=w^{-1} h\a\},i=1,2.$$ Then $\tau(M_1)=M_2$. Note that $x,y\in M_1$. Thus $M_1, h$
  generate $\mathcal{H}(f_1)$, from which we have $M_2=\tau(M_1), h=\tau(h)$ generates
  $\mathcal{H}(f_2)$. Using Lemma 1 it is straightforward to verify that
 $$M_2=\span\{x^ih^ky^j|i,j,k\in \Z_+, w^{p(j-i)}\in \{w,w^{-1}\}\}.$$
 Now suppose to the contrary that $w^{p}\not\in \{w,w^{-1}\}$. Then the
 subalgebra generated by $M_2, h$ is contained in the proper subalgebra
 $\C+\span\{x^ih^ky^j|i,j,k\in \Z_+, i+j\ge 2\}$ of $\mathcal{H}(f_2)$, a
 contradiction. So $a_2=a_1^{\pm1}$.

 \begin{claim} If $\deg f_1\ge 2$, then $f_2(h)=af_1(a^{-1}(h-c))+c$ for some $a\in \C^*, c\in \C$.\end{claim}

 {\it Proof of the Claim 4}. From Claim 2, we have $\deg f_2\ge 2$. From Lemma \ref{lemma-2},
 $\tau(h)\tau(x)=\tau(x)f_1(\tau(h))$, we have $\deg(\tau(h))=(0,0)$, i.e., $\tau(h)\in \C[h]$.
  Similarly we have $\tau^{-1}(h)\in \C[h]$. Thus $\tau(\C[h])=\C[h]$. So we have $\tau(h)=a^{-1}(h-c)$
  for some $a\in \C^*$ and $c\in \C$. Using the isomorphism $\sigma:\mathcal{H}(f_2)\rightarrow \mathcal{H}(f'_2)$, where $f_2(ah+c)=af'_2(h)+c$, in the proof of
  Lemma \ref{lemma-3} (1), we may assume that $\tau(h)=h$ with $f_2(h)$ being replaced by
  $f'_2(h)=a^{-1}(f_2(ah+c)-c)$.

  From Theorem \ref{thm-4}, we have $\tau(\C[z])=\C[z]$. Thus
   $\tau(z)=a'z+c'$ for some $a'\in \C^*$ and $c'\in \C$. So we have $\tau(xy-h)=a'(xy-h)+c'$, i.e.,
    $\tau(x)\tau(y)=a'xy+(1-a')h+c'$. Thus $$\deg \tau(x)+\deg \tau(y)=(1,1),$$ $$\deg \tau(x)\ne
    (0,0,
    \text{ and } \deg \tau(y)\ne(0,0).$$

   \begin{case}$\deg \tau(x)=(1,0)$ and $\deg \tau(y)=(0,1)$.\end{case}

   In this case, we may assume that $\tau(x)=xg_1(h)+g_2(h,y)$ and $\tau(y)=g_3(h)y+g_4(h)$.
   From $$a'xy+(1-a')h+c'=\tau(x)\tau(y)$$ $$= xg_1(h)g_3(h)y+
   xg_1(h)g_4(h)+g_2(h,y)g_3(h)y+g_2(h,y)g_4(h),$$
   using Lemma 1 it is easy to see that
 $g_1(h)g_3(h)=a'$, $g_2(h,y)=g_4(h)=0$, $a'=1,c'=0$ i.e.
 $$\tau(x)=b x, \tau(y)=b^{-1}y,$$ for some $b\in\C^*$. Therefore from $\tau(hx)=\tau(xf_1(h))=bxf_1(h)$ and
 $\tau(hx)=\tau(h)\tau(x)=hbx=bxf'_2(h)$, we have $f_1(h)=f'_2(h)$, i.e.,
 $f_2(h)=af_1(a^{-1}(h-c))+c$ for some $a\in \C^*,
c\in \C$.

  \begin{case}$\deg \tau(x)=(0,1)$ and $\deg \tau(y)=(1,0)$.\end{case}

    In this case, we may assume that $\tau(y)=xg_1(h)+g_2(h,y)$ and $\tau(x)=g_3(h)y+g_4(h)$.
    From $$a'(xy-h)+c'+f_1(h)=\tau(y)\tau(x)$$ $$= xg_1(h)g_3(h)y+ xg_1(h)g_4(h)+g_2(h,y)g_3(h)y+g_2(h,y)g_4(h),$$
    using Lemma 1
    it is easy to see that
 $g_1(h)g_3(h)=a'$, $g_2(h,y)=g_4(h)=0$, $-a'h+f_1(h)+c'=0$, which cannot occur since $\deg(f_1(h))\ge 2$.

 This completes the proof.
\end{proof}

\section{Fnite dimensional simple modules}

In this section, we will determine all finite dimensional simple
modules over $\H(f)$ for any polynomial $f(h)\in\C[h]$. For any
$f\in \C[h]$, define
$$S_f=\{\lambda:\Z\rightarrow \C| f(\lambda(i))=\lambda(i+1),\forall
i\in \Z\}.$$ It is clear that for any $\lambda\in S_f$, the set
$$\{m| \lambda(i)=\lambda(m+i),\forall i\in \Z\}$$ is an additive
subgroup of $\Z$. Denote by $|\lambda|$ the nonnegative integer with
$$\Z|\lambda|=\{m| \lambda(i)=\lambda(m+i),\forall i\in \Z\}.$$

The following result is crucial in this section.
\begin{lemma} Let $\lambda\in S_f$ with $|\lambda|=m\ne 0$. Then  $\lambda{(i)}=\lambda{(j)}$ for some $i,j\in\Z$
if and only if $m|i-j$.  \end{lemma}

\begin{proof} If $m|i-j$, by definition we see that
$\lambda{(i)}=\lambda{(j)}$.

Suppose that $\lambda{(i)}=\lambda{(j)}$ for some $i,j\in\Z$  and $m
\nmid i-j$. We may assume that $i<j$ and $j-i=mq+r$ with $q\in\Z_+$
and $0<r<m$. Then we see that
$\lambda{(i+mq)}=\lambda{(i)}=\lambda{(j)}=\lambda{(i+mq+r)}$. Let
$i_1=i+mq$. We know that $\lambda{(i_1)}=\lambda{(i_1+r)}$,
consequently, $$\lambda{(k)}=\lambda{(k+r)}, \forall k\ge i_1.$$
Since $|\lambda|=m$, there exists $i_2\in\Z$ such that
$\lambda{(i_2)}\not=\lambda{(i_2+r)}$. There exists $q_1\in\N$ such
that $i_2+mq_1>i_1$. So $\lambda{(i_2+mq_1)}=\lambda{(i_2+mq_1+r)}$.
But   $\lambda{(i_2)}=\lambda{(i_2+mq_1+r)}$ and
$\lambda{(i_2+r)}=\lambda{(i_2+mq_1+r)}$, which contradict the
assumption that $\lambda{(i_2)}\not=\lambda{(i_2+r)}$. The lemma has
to be true.
\end{proof}

We remark that for $\lambda\in S_f$ with $|\lambda|=0$, it is
possible that $\lambda(i)=\lambda(j)$ for some different $i,
j\in\Z$.

 For any $\lambda\in S_f$ and $\dot{z}\in \C$, define
the action of $\mathcal{H}(f)$ on the vector space
$A_{\mathcal{H}(f)}(\lambda,\dot{z})=\C[t,t^{-1}]$ by
$$h t^i=\lambda(i)t^i, \,\,\,x t^i= t^{i+1},\,\,\, yt^i=(\lambda(i)+\dot{z})t^{i-1}.$$

Similarly we define the action of $\mathcal{H}(f)$ on
$B_{\mathcal{H}(f)}(\lambda,\dot{z})=\C[t,t^{-1}]$ by
$$h t^i=\lambda(i)t^i,\,\,\, x t^i=(\lambda(i+1)+\dot{z})t^{i+1},\,\,\,y t^i= t^{i-1}. $$

\begin{lemma} \begin{itemize}
\item[(1).] $A_{\mathcal{H}(f)}(\lambda,\dot{z})$ and  $B_{\mathcal{H}(f)}(\lambda,\dot{z})$ are $\mathcal{H}(f)$-modules.
\item[(2).] The $\mathcal{H}(f)$-module $A_{\mathcal{H}(f)}(\lambda,\dot{z})$ (resp.
$B_{\mathcal{H}(f)}(\lambda,\dot{z})$) is simple if and only if
$\lambda(i)+\dot{z}\ne 0$ for all $i\in \Z$ and $|\lambda|=0$.
\item[(3).] Suppose that $\lambda(i)+\dot{z}\ne 0$ for all $i\in \Z$ and
$m=|\lambda|\ne 0$. Then any nonzero submodule of
$A_{\mathcal{H}(f)}(\lambda,\dot{z})$ (resp.
$B_{\mathcal{H}(f)}(\lambda,\dot{z})$) is equal to $\C[t,t^{-1}]
g(t)$ for some weight vector $g(t)\in \C[t^{m}]$. In particular, any
maximal nonzero submodule of $A_{\mathcal{H}(f)}(\lambda,\dot{z})$
is equal to $\C[t,t^{-1}](t^m-a)$ for some $a\in \C^*$.\end{itemize}
  \end{lemma}

\begin{proof} (1).  In $A_{\mathcal{H}(f)}(\lambda,\dot{z})$, we have
$$\aligned &hx t^i=ht^{i+1}=\lambda(i+1)t^{i+1},\\
&xf(h) t^i=xf(\lambda(i))t^i=\lambda(i+1)t^{i+1};\\
&yh t^i=y\lambda(i)t^i=\lambda(i)(\lambda(i)+\dot{z})t^{i-1},\\
&f(h)y t^i=f(h)(\lambda(i)+\dot{z})t^{i-1}=\lambda(i)(\lambda(i)+\dot{z})t^{i-1};\\
&(yx-xy)t^i=(\lambda(i+1)-\lambda(i))t^{i}=(f(h)-h)t^i.\endaligned$$
 Hence from the definition of $\mathcal{H}(f)$, $A_{\mathcal{H}(f)}(\lambda,\dot{z})$ is an $\mathcal{H}(f)$-module.

In $B_{\mathcal{H}(f)}(\lambda,\dot{z})$, we have
$$\aligned &hx t^i=(\lambda(i+1)+\dot{z})ht^{i+1}=(\lambda(i+1)+\dot{z})\lambda(i+1)t^{i+1},\\
&xf(h) t^i=xf(\lambda(i))t^i=\lambda(i+1)xt^{i}=(\lambda(i+1)+\dot{z})\lambda(i+1)t^{i+1};\\
&yh t^i=y\lambda(i)t^i=\lambda(i)t^{i-1},\\
&f(h)y t^i=f(h)t^{i-1}=f(\lambda(i-1))t^{i-1}=\lambda(i)t^{i-1};\\
&(yx-xy)t^i=(\lambda(i+1)-\lambda(i))t^{i}=(f(h)-h)t^i.\endaligned$$
 Hence $B_{\mathcal{H}(f)}(\lambda,\dot{z})$ is an $\mathcal{H}(f)$-module.

(2). It is easy to see that $t^i\C[t]$ is a submodule of
$A_{\mathcal{H}(f)}(\lambda,\dot{z})$ if $\lambda(i)+\dot{z}=0$, and
$(t^m-1)\C[t, t^{-1}]$ is a submodule of
$A_{\mathcal{H}(f)}(\lambda,\dot{z})$ if $|\lambda|=m\ne0$.

Now suppose that $|\lambda|=0$ and $\lambda(i)+\dot{z}\ne 0$ for
all $i\in \Z$. Let $V$ be any nonzero submodule of
$A=A_{\mathcal{H}(f)}(\lambda,\dot{z})$. Take $0\ne
g(t)=t^k\sum_{i=0}^s a_i t^i\in V$ with minimal $s$.  If $s=0$, then
$g(t)=a_0t^k$, which implies $\mathcal{H}(f) t^k=A$. Now suppose
that $s>0$.  From $hg(t)=\sum_{i=0}^s a_i\lambda(k+i) t^{k+i}$ and the minimality of $s$,
 we have $\lambda(k)=\lambda(k+i)$ for all $i$ with $a_i\ne 0$, i.e. $g(t)$ is a weight vector. It is clear that  $\C[t]g(t)=\C[x] g(t)\subset V$. Since $g(t)$ is a
weight vector, we have $0\ne yg(t)=(\lambda(k)+\dot{z})\sum_{i=0}^s
a_i t^{k-1+i}\in V$, i.e. $t^{-1}g(t)\in V$, which from the minimal
of $s$ is also a weight vector. Inductively we have $t^{-i}g(t)\in
V,\forall i\in \N$ and $t^{-i}g(t)$ are weight vectors for all $i\in
\Z$. In particular we have $\lambda(i)=\lambda(s+i),\forall i\in \Z$, a contradiction with $|\lambda|=0$. Thus $A_{\mathcal{H}(f)}(\lambda,\dot{z})$
is a simple $\mathcal{H}(f)$-module. The same arguments are valid
for $B_{\mathcal{H}(f)}(\lambda,\dot{z})$.

(3). Now suppose that $|\lambda|=m\ne0$ and $\lambda(i)+\dot{z}\ne
0$ for all $i\in \Z$. Let $V$ be any nonzero submodule of
$A=A_{\mathcal{H}(f)}(\lambda,\dot{z})$. Take $0\ne
g(t)=t^k\sum_{i=0}^s a_i t^i\in V$ with minimal $s$. If $s=0$, then
$g(t)=a_0t^k$, which implies $\mathcal{H}(f) t^k=A$. Now suppose
that $s>0$. From the same argument as in Part (2), we have $\C[t,t^{-1}]g(t)\subset V$ and $t^ig(t)$ are weight vectors for all $i\in \Z$.
In particular, we have $\lambda(j)=\lambda(j+i)$ for all $i$ with $a_i\ne 0$ and $j\in \Z$. Using Lemma 6 we deduce that
  $g(t)\in\C[t^m,t^{-m}]$.

And for any $0\ne v\in V$, from the minimality of
$s$, it is easy to see that $g(t)|v$ in the Laurient polynomial
algebra $\C[t,t^{-1}]$. Therefore $\C[t,t^{-1}]g(t)=V.$ Similarly we have  (3) for
$B_{\mathcal{H}(f)}(\lambda,\dot{z})$. \end{proof}

Now we are going to construct and  classify all finite dimensional
simple modules over $\mathcal{H}(f)$. For any $\lambda\in
S_f,\dot{z}\in \C,a\in \C^*$ with $|\lambda|=m\ne 0$, we have the
$m$-dimensional $\mathcal{H}(f)$-module
$$A'_{\mathcal{H}(f)}(\lambda,\dot{z},a)=A_{\mathcal{H}(f)}
(\lambda,\dot{z})/\C[t,t^{-1}](t^m-a);$$
$$B'_{\mathcal{H}(f)}(\lambda,\dot{z},a)=B_{\mathcal{H}(f)}(\lambda,\dot{z})/\C[t,t^{-1}](t^m-a).$$

\begin{lemma}\label{lemma-AB} The $\mathcal{H}(f)$-modules $A'_{\mathcal{H}(f)}(\lambda,\dot{z},a)$,
$B'_{\mathcal{H}(f)} (\lambda,\dot{z},a)$ are simple.\end{lemma}
\begin{proof}Note that $A'_{\mathcal{H}(f)}(\lambda,\dot{z},a)$ is a weight module with the support set
 $\{\lambda(i)|i=0,1,2,\ldots,m-1\}$ and each nonzero weight space is of dimension 1. Hence any nonzero
 submodule must contain a weight vector $t^i+\C[t,t^{-1}](t^m-a)$ for some $i$, which by the repeated
 actions
 of $x$, can generate the whole module. So $\mathcal{H}(f)$-modules $A'_{\mathcal{H}(f)}(\lambda,\dot{z},a)$ is simple.
 The argument for $B'_{\mathcal{H}(f)} (\lambda,\dot{z},a)$ is similar.\end{proof}

\begin{lemma}\label{xV=V}Let $V$ be any simple module over $\mathcal{H}(f)$ with $\dim V=n$.
If $x^n V\ne 0$, then $V\cong
A'_{\mathcal{H}(f)}(\lambda,\dot{z},a)$ for some $\dot{z}\in \C,
a\in \C^*$ and $\lambda\in S_f$ with $|\lambda|=n$.\end{lemma}

\begin{proof}Let $W=\cap_{i=1}^{\infty} x^i V=x^n V$. Then $W\ne 0$  and $x$ acts bijectively on $W$.
It is easy to see that $h W\subset W$. Hence there exist some $0\ne
v\in W$ and $b\in \C$ with $h v=b v$. Then $h x^i v=x^if^{(i)}(h)
v=f^{(i)}(b) x^i v$.  Therefore there exists some $k\in \Z_+, m\in
\N$ such that $f^{(k)}(b)=f^{(k+m)}(b)$. We may assume that this $m$
is minimal. Replace $v$ by $x^k v$ and $b$ by $f^{(k)}(b)$. Then
$\C[x^m]v$ is contained in the eigenvector space of $h$ with respect
to the eigenvalue $b$. We may assume that $0\ne v\in W$ satisfies
 $h v=b v$, $x^m v=a v$ for some $b\in \C, a\in \C^*$ with $f^{(m)}(b)=b$.
 Let $W'=\sum_{i=0}^{m-1}\C x^i v$. Then $$x W'=W',\,\,\,hW'\subset W',\,\,\,yW'=(yx)W'=(\dot{z}+f(h))W'\subset
 W'.$$
 Thus $W'=V$ and $m=n$. And it is straightforward to verify that $V\cong A'_{\mathcal{H}(f)}(\lambda,\dot{z},a)$,
 where $\lambda(0)=b,\lambda(i)=f^{(i)}(b),$ for all $i\in \N$ and $\lambda(i)=\lambda(i+n)$ for all $i\in \Z$.
 \end{proof}

\begin{lemma}\label{yV=V}Let $V$ be any simple module over $\mathcal{H}(f)$ with $\dim V=n$. If $yV=V$ and $x^n V=0$,
then $V\cong B'_{\mathcal{H}(f)}(\lambda,\dot{z},a)$ for some $\dot{z}\in \C, a\in \C^*$, $|\lambda|=n$
and $\lambda(i)+\dot{z}=0$ for some $i$.\end{lemma}

\begin{proof} Let $\sigma$ be the representation associated to $V$. Then $\sigma(y)$ is invertible.
And from $\sigma(y)\sigma(h)=\sigma(f(h))\sigma(y)$, we have
$\sigma(h)\sigma(y)^{-1}=\sigma(y)^{-1}\sigma(f(h))$. Since $V$ is
finite dimensional, there exist some $0\ne v\in V$ and $b\in \C$
with $h v=b v$. From $\sigma(h)\sigma(y)^{-i}
(v)=\sigma(y)^{-i}f^{(i)}(h)(v)=f^{(i)}(b)\sigma(y)^{-i}(v)$, we
have $f^{(k)}(b)=f^{(k+m)}(b)$ for some $k\in \Z_+, m\in \N$. We may
assume that this $m$ is minimal. Replace $v$ by $\sigma(y)^{-k} v$
and $b$ by $f^{(k)}(b)$. Then $\C[\sigma(y)^{-m}]v$ is contained in
the eigenvector space of $h$ with respect to the eigenvalue $b$. Now
we may assume that $0\ne v\in V$ satisfies $h v=b v$,
$\sigma(y)^{-m} v=a v$ for some $b\in \C, a\in \C^*$ with
$f^{(m)}(b)=b$. Let $$W=\sum_{i=0}^{m-1}\C y^i v=\sum_{i=0}^{m-1}\C
\sigma(y)^{-i} v.$$ Then $y W=W$, $hW\subset W$,
$xW=(xy)W=(\dot{z}+h)W\subset W$. Thus $W=V$ and $m=n$. And it is
straightforward to verify that $V\cong
B'_{\mathcal{H}(f)}(\lambda,\dot{z},a)$, where
$\lambda(0)=b,\lambda(i)=f^{(i)}(b),\forall i\in \N$ and
$\lambda(i)=\lambda(i+n)$ for all $i\in \Z$. And from $x^n V=0$, we
have $\lambda(i)+\dot{z}=0$ for some $i$.
\end{proof}
%Suppose that $yv=0$ for some $0\ne v\in V$ and $x^n V=0$. Then $hv=-\dot{z}v$. And from Lemma 1 $V=\C[x]v=\sum_{i=0}^{n-1} x^i v$. $hx^iv=x^if^{(i)}(-\dot{z})v=f^{(i)}(-\dot{z})x^i v$. $yx^i v=(yx)x^{i-1}v=(f(h)+\dot{z})x^{i-1}v=(\dot{z}+(f^{(i)})(-\dot{z}))x^{i-1}v$.

For any $n\in \N$ and $\dot{z}\in \C$ with
$\dot{z}+f^{(n)}(-\dot{z})=0$, define the action of $\mathcal{H}(f)$
on $C_{\mathcal{H}(f)}(\dot{z},n)=\C[t]/(t^n)$ by
$$ y t^i=(\dot{z}+f^{(i)}(-\dot{z}))t^{i-1},\forall i=0,1,\ldots,n-1,$$
$$h t^i=f^{(i)}(-\dot{z})t^i, \,\,\,x t^i=t^{i+1},\forall i=0,1,\ldots,n-1.$$

\begin{lemma}\label{lemma-C} \begin{itemize}\item[(1).] For any $n\in \N$ and $\dot{z}\in \C$
with $\dot{z}+f^{(n)}(-\dot{z})=0$, $C_{\mathcal{H}(f)}(\dot{z},n)$
is an $\mathcal{H}(f)$-module.
\item[(2).] The $\mathcal{H}(f)$-module $C_{\mathcal{H}(f)}(\dot{z},n)$ is simple if and only if
$\dot{z}+f^{(i)}(-\dot{z})$ $\ne 0,$ for all
$i=1,2,\ldots,n-1$.\end{itemize}\end{lemma}
\begin{proof} (1). It is straightforward to compute out that

 $$\aligned &hx t^i=h t^{i+1}=(f^{i+1})(-\dot{z})t^{i+1}, xf(h)t^i=(f^{i+1})(-\dot{z})t^{i+1};\\
 &yh t^i=yf^{(i)}(-\dot{z}) t^i=f^{(i)}(-\dot{z})(\dot{z}+f^{(i)}(-\dot{z}))t^{i-1},\\
 & f(h)yt^i=f^{(i)}(-\dot{z})(\dot{z}+f^{(i)}(-\dot{z}))t^{i-1};\\
 &yx t^i=(\dot{z}+f^{(i+1)}(-\dot{z}))t^{i},i=0,\ldots,n-2,\\
 &xy t^i=(\dot{z}+f^{(i)}(-\dot{z}))t^i,i=0,\ldots,n-1.\endaligned$$
Hence $(yx-xy)t^i=(f(h)-h)t^i,\forall i=0,1,\ldots,n-2$. From
$\dot{z}+f^{(n)}(-\dot{z})=0$, and
$(yx-xy)t^{n-1}=-(\dot{z}+(f^{n-1})(-\dot{z}))t^{n-1}$, we have
$$(f(h)-h)t^{n-1}=f^{(n)}(-\dot{z})-f^{(n-1)}(-\dot{z})
t^{n-1}=(yx-xy)t^{n-1}.$$ Thus from the definition of
$\mathcal{H}(f)$, $C_{\mathcal{H}(f)}(\dot{z},n)$ is an
$\mathcal{H}(f)$-module.

(2). The ``if part" follows easily from the action of $x$ and $y$.
If $\dot{z}+f^{(i)}(-\dot{z})=0$ for some $i\in \{1,2,\ldots,n-1\}$. Then
it is straightforward to verify that
$\span\{t^i,t^{i+1},\ldots,t^{n-1}\}$ forms a proper submodule.
 \end{proof}

Now we are ready to give the classification of all
finite-dimensional simple modules over $\mathcal{H}(f)$.

\begin{theorem} Any $n$-dimensional simple $\mathcal{H}(f)$-module $V$ is isomorphic to only
one of the following simple modules up to an index shift of
$\lambda$:
\begin{itemize}
\item[(1).]  $A'_{\mathcal{H}(f)}(\lambda,\dot{z},a)$ for some $\lambda\in
S_f,\dot{z}\in \C, a\in \C^*$ with $|\lambda|=n$;
\item[(2).]   $B'_{\mathcal{H}(f)}(\lambda,\dot{z},a)$ for some $\dot{z}\in \C, a\in \C^*$, $|\lambda|=n$ with $\lambda(i)+\dot{z}=0$ for some $i$;
\item[(3).]   $C_{\mathcal{H}(f)}(\dot{z},n)$ for some $\dot{z}\in \C$ with
$\dot{z}+f^{(n)}(-\dot{z})=0$ and $\dot{z}+f^{(i)}(-\dot{z})\ne
0,\forall i=1,2,\ldots,n-1$.\end{itemize}
\end{theorem}

\begin{proof}From Lemma \ref{lemma-AB} and Lemma \ref{lemma-C}, we have the simplicity of the modules in (1)-(3).
From Lemma \ref{xV=V} and Lemma \ref{yV=V}, we only need to consider
the case $yV\ne V$ and $x^n V=0$. Note that $V_1=\{v\in V| yv=0\}$
is a nonzero $\C[h]$ submodule of $V$. We may choose a nonzero $v\in
V_1$ with $h v=b v$ and $y v=0$. From Lemma 1 we know that $V=\C[x]
v$. Suppose that $z$ acts as the scalar $\dot{z}$. Then $\dot{z}
v=(xy-h)v=-b v$. Thus $b=-\dot{z}$. Now it is straightforward to
verify $V\cong C_{\mathcal{H}(f)}(\dot{z},n)$. The rest of the
theorem is clear. So we have proved the theorem.
\end{proof}

\begin{example}(1). If $f=0$, $\lambda\in S_f$  with
$|\lambda|<\infty$, then $\lambda(i)=0$ for all $i\in\Z$.  Any
finite dimensional simple $\mathcal{H}(0)$ module must be isomorphic
to $\C v$ with the action $x v=a v, hv=0, yv=bv$ for some $a, b\in
\C$.

(2). Let $f=h+c$ for some $c\in \C^*$. Since $|\lambda|=\infty$ for
any $\lambda\in S_f$,  there is no finite dimensional simple
$\mathcal{H}(f)$ module.

(3).  Let $f=ah$, where $a\in \C^*$ is not a root of unity. If
$\lambda\in S_f$  with $|\lambda|<\infty$, then $\lambda(i)=0$ for
all $i\in\Z$. Any
 finite dimensional simple $\mathcal{H}(f)$ module must be
isomorphic to $\C v$ with the action $x v=a v, hv=0, yv=bv$ for some $a, b\in
\C$.

(4). Let $f=wh$, where $w$ is an $n$-th ($n>1$) primitive root of
unity. Then any finite dimensional simple $\mathcal{H}(f)$ module
must be isomorphic to one of the following modules

(a). $\C v$ with the action $x v=a v, hv=0, yv=bv$ for some $a, b\in
\C$.

(b). $A'(\lambda,\dot{z}, a)=\C[t]/(t^n-a)$ with $\lambda(i)=bw^i,
b\in \C^*,a\in \C^*$ and the action
$$h t^i=b w^it^i, x t^i= t^{i+1}, yt^i=(bw^i+\dot{z})t^{i-1}.$$

(c).  $B'(\lambda,\dot{z}, a)=\C[t]/(t^n-a)$ with $\lambda(i)=bw^i,
b\in \C^*,a\in \C^*, \dot{z}=-bw^j$ for some $j$ and the action
$$h t^i=b w^it^i, x t^i=b(w^{i+1}-w^j)t^{i+1},y t^i= t^{i-1}. $$

(d). $C(\dot{z},n)=\C[t]/t^n$ for some $\dot{z}\in \C^*$ with the
action
$$ y t^i=(1-w^i)\dot{z}t^{i-1}, h t^i=-w^i\dot{z}t^i, x t^i=t^{i+1},\forall i=0,1,\ldots,n-1.$$
From Lemma 4, we know that all irreducible modules over
$\mathcal{H}(wh)$ are finite dimensional. Thus we have the
classification of all irreducible modules over $\mathcal{H}(wh+c)$
for any $c\in\C$ and  for any $n$-th ($n>1$) primitive root $w$ of
unity.
 \end{example}

\begin{example}Let $m\in \N+1$ and $f=h^m$. Then any finite dimensional simple
$\mathcal{H}(f)$ module must be isomorphic to one of the following
modules

(a). $\C v$ with the action $x v=a v, hv=cv,  yv=bv$ for some $a, b,
c\in \C$ with $c^m=c$.

(b). $A'(\lambda,\dot{z}, a)=\C[t]/(t^n-a)$ with the action
$$h t^i= w^{m^i}t^i, x t^i= t^{i+1}, yt^i=(w^{m^i}+\dot{z})t^{i-1},\forall i=0,1,\ldots,n-1,$$ for some
$a\in \C^*,\dot{z}\in \C,  n\in \N+1$, and  $w$ a $({m^n-1})$-th root of unity such that $w^{m^i-1}\ne 1$ for all $i=1,2\ldots,n-1$.

(c).  $B'(\lambda,\dot{z}, a)=\C[t]/(t^n-a)$ with the action
$$h t^i=w^{m^i}t^i, x t^i=(w^{m^{i+1}}-w^{m^j})t^{i+1},y t^i= t^{i-1},\forall i=0,1,\ldots,n-1,$$ for some
$a\in \C^*,j\in \Z_+, n\in \N+1$, and  $w$ a $({m^n-1})$-th root of unity such that $w^{m^i-1}\ne 1$ for all $i=1,2\ldots,n-1$.

(d). $C(\dot{z},n)=\C[t]/t^n$ with $n\in \N+1$, where
$(-\dot{z})^{m^n-1}=1$  and $(-\dot{z})^{m^i-1}\ne 1$ for all
$i=1,2,\ldots,n-1$, and with the action
$$ y t^i=(\dot{z}+(-\dot{z})^{m^i})t^{i-1},h t^i=(-\dot{z})^{m^i}t^i, \,\,\,x t^i=t^{i+1},\forall i=0,1,\ldots,n-1.$$
 \end{example}

The following example shows  an interesting property of many
$\mathcal{H}(f)$.

\begin{example}  Let  $f=h^m$ with $m>1$. It is easy to see that, for any $n\in\N$,
there is a $c\in\C$ such that $f^{(n)}(c)=c$ and $f^{(i)}(c)\ne c$
for any $0<i<n$. Using Theorem 12 we know that  there are infinitely
many simple $n$-dimensional $\mathcal{H}(f)$-module. Then
$\mathcal{H}(f)$ has infinitely many ideals $I_n$  such that
$\mathcal{H}(f)/I_n\cong M_n(\C)$ for any $n\in\N$.

This is not always true. If we take $f(h)=h^2+2h-3/4$, then $f(h)-h=(h+3/2)(h-1/2)$ and $f^{(2)}(h)-h=(h+3/2)^3(h-1/2)$. Thus we do
not have $\lambda\in S_f$ with $|\lambda|=2$, i.e., $\mathcal{H}(f)$
does not have a simple $2$-dimensional module.
\end{example}

We conclude this paper by an open question on generalized Heisenberg
algebras  $\mathcal{H}(f)$: How to determine all simple weight
modules for $\mathcal{H}(f)$? How about all simple modules over
$\mathcal{H}(f)$?

\

\noindent
{\bf Acknowledgments.}
The research in this paper was carried out during the visit of the first author
to University of Waterloo and  to Wilfrid Laurier University.
K.Z. is partially supported by  NSF of China (Grant
11271109), NSERC and University Research Professor grant at Wilfrid Laurier University. R.L. is partially supported by NSF of China
(Grant 11371134) and Jiangsu Government Scholarship for Overseas Studies (JS-2013-313).
R.L.  would like to thank professors Wentang Kuo
and   Kaiming Zhao for sponsoring his visit, and University of Waterloo
for providing excellent working conditions.

\vskip 10pt

 \noindent \noindent R.L.: Department of Mathematics,
Soochow University,  Suzhou,  P. R. China; e-mail:
{rencail\symbol{64}amss.ac.cn}

\

\noindent K.Z.: Department of Mathematics, Wilfrid Laurier
University, Waterloo, Ontario, N2L 3C5, Canada; and College of
Mathematics and Information Science, Hebei Normal (Teachers)
University, Shijiazhuang 050016, Hebei, P. R. China; e-mail:
{kzhao\symbol{64}wlu.ca}

\begin{thebibliography}{99999}

\bibitem[ADF]{ADF} M. Angelova, V. K. Dobrev, A. Frank, Simple applications of q-bosons, J. Phys. A 34 (2001) L503-L509.
\bibitem[AEGPL]{AEGPL} S. S. Avancini,  A. Eiras,  D. Galetti,  B. M. Pimentel, C. L. Lima, Phase-transition in a q-deformed lipkin model, J. Phys. A: Math. Gen., 28 (1995), 4915-4923.
\bibitem[Ba]{Ba}   V. V.  Bavula,  Generalized Weyl algebras and their representations. (Russian)  Algebra i Analiz  4  (1992),  no. 1, 75-97;  translation in  St. Petersburg Math. J.,  4  (1993),  no. 1, 71-92.
\bibitem[BJ]{BJ}   V. V.  Bavula,  D. A.
Jordan,  Isomorphism problems and groups of automorphisms for generalized Weyl algebras. Trans. Amer. Math. Soc. 353 (2001), no. 2, 769-794.
\bibitem[BEH]{BEH} K. Berrada, M. El Baz, Y.  Hassouni, Generalized Heisenberg algebra coherent states for power-law potentials, Physics Letter  A, 375 (2011),   3,     298-302.
 \bibitem[BCR1]{BCR1} V.B. Bezerra, E.M.F. Curado, M.A. Rego-Monteiro, Ground-state entropies of the Potts antiferromagnet on diamond hierarchical lattices, Phys. Rev. D, 66
(2002), 085013.
   \bibitem[BCR2]{BCR2} V.B. Bezerra, E.M.F. Curado, M.A. Rego-Monteiro, Renormalization of a deformed scalar quantum field theory, Phys. Rev. D, 69
(2004), 085003.
\bibitem[B]{B} L. Biedenharn, The quantum group suq(2) and a q-analogue of the boson operators, J. Phys. A, 22 (1989), L873-L878.
\bibitem[Bl]{Bl} R.~Block. The irreducible representations of the Lie algebra $\mathfrak{sl}(2)$ and of the Weyl
algebra. Adv. Math. {\bf 139} (1981), no. 1, 69--110.
\bibitem[BD]{BD} D. Bonatsos, C. Daskaloyannis, Quantum groups and their applications in nuclear physicsm, Prog. Part. Nucl. Phys., 43 (1999), 537-618.
\bibitem[CHR]{CHR} E. M. F. Curado,  Y. Hassouni,  M. A. Rego-Monteiro, Generalized Heisenberg algebra and algebraic method: The example of an infinite square-well potential, Physics Letter A,  372(2008), 19,    3350-3355.
\bibitem[CR1]{CR1} E. M. F. Curado,    M. A. Rego-Monteiro,  Thermodynamic properties of a solid exhibiting the energy spectrum given by the logistic map, Phys. Rev. E, 61(2000), 6255
\bibitem[CR2]{CR2} E. M. F. Curado,    M. A. Rego-Monteiro, Multi-parametric deformed Heisenberg algebras: a route to complexity,
J. Phys. A: Math. Gen., 34 (2001), 3253-3264.
\bibitem[CRRH]{CRRH} E.M.F. Curado, M.A. Rego-Monteiro, L.M.C.S. Rodrigues, Y. Hassouni, Coherent states for a degenerate system: The hydrogen atom,
Physica A 371 (2006), 16-19.
\bibitem[CRRL]{CRRL} E. M. F. Curado,  M. A. Rego-Monteiro,  Rodrigues,  M. C. S. Ligia,
Structure of generalized Heisenberg algebras and quantum decoherence
analysis,  Physical Review A.  Volume: 87(2013), 5,      052120.
\bibitem[DCR]{DCR} J. de Souza, E.M.F. Curado, M.A. Rego-Monteiro, Generalized Heisenberg algebras and Fibonacci series, J. Phys. A, 39 (2006),
10415-10425.
\bibitem[DOR]{DOR} J. de Souza,  N. M. Oliveira-Neto,  C. I. Ribeiro-Silva, A method based on a nonlinear generalized Heisenberg algebra to study the molecular vibrational spectrum, European physical journal D, 40 (2006),  2,     205-210.
\bibitem[GPLL]{GPLL}  D. Galetti,  B. M. Pimentel, C. L. Lima,   J. T. Lunardi, Q-deformed fermionic Lipkin model at finite temperature, Physica A, 242(1997), 501-508.
     \bibitem[HCR]{HCR} Y. Hassouni, E.M.F. Curado, M.A. Rego-Monteiro, Construction of coherent states for physical algebraic systems, Phys. Rev. A, 71
(2005), 022104;
    \bibitem[M]{M} A. J. Macfarlane, On q-analogues of the quantum harmonic oscillator and the quantum group SU(2)q,  J. Phys. A, 22 (1989) 4581-4588.
\bibitem[MRW]{MRW} M. R. Monteiro, L. M. C. S. Rodrigues, S. Wulk, Quantum algebraic nature of the phonon spectrum in He-4, Phys. Rev. Lett., 76 (1996), 1098-1101.
\bibitem[P]{P}  M. S. Plyushchay, Deformed Heisenberg algebra with reflection,   Nucl. Phys. B, 491(1997), 619-634.
\bibitem[S]{S}  S. P. Smith, A class of algebras similar to the enveloping algebra of $\mathfrak{sl}(2)$,
 Trans. Amer. Math. Soc., 322 (1990), no. 1, 285-314.
 \end{thebibliography}
\end{document}